\documentclass[conference,10pt]{IEEEtran}
\IEEEoverridecommandlockouts
\usepackage{amsmath,amsfonts,amsthm,amssymb,bbm}
\usepackage{thm-restate}
\usepackage{cite}
\usepackage{balance}
\usepackage{mathrsfs}
\usepackage{xcolor}

\DeclareMathOperator*{\argmax}{arg\,max}

\theoremstyle{plain} 
\newtheorem{theorem}{Theorem}

\newtheorem{corollary}{Corollary}
\newtheorem{definition}{Definition}
\newtheorem{lemma}{Lemma}

\theoremstyle{definition} \newtheorem{remark}{Remark}
\theoremstyle{definition}

\title{Maximal Guesswork Leakage\thanks{The work of M. Managoli and V. Prabhakaran was supported by DAE under project no. RTI4001. V. Prabhakaran was additionally supported by SERB through project MTR/2020/000308.}}
\author{%
\IEEEauthorblockN{Gowtham R. Kurri\IEEEauthorrefmark{1},
                Malhar A. Managoli\IEEEauthorrefmark{2},
                Vinod M. Prabhakaran\IEEEauthorrefmark{2}
                    }

                    \IEEEauthorblockA{\IEEEauthorrefmark{1}%
                    International Institute of Information Technology, Hyderabad, India, \texttt{gowtham.kurri@iiit.ac.in}}

                    \IEEEauthorblockA{\IEEEauthorrefmark{2}%
                    Tata Institute of Fundamental Research, Mumbai, India, \texttt{\{malhar.managoli,vinodmp\}@tifr.res.in}}

}

\def \extended {1}

\begin{document}
\maketitle
\begin{abstract}
    We introduce the study of information leakage through \emph{guesswork}, the minimum expected number of guesses required to guess a random variable. In particular, we define \emph{maximal guesswork leakage} as the multiplicative decrease, upon observing $Y$, of the guesswork of a randomized function of $X$, maximized over all such randomized functions. We also study a pointwise form of the leakage which captures the leakage due to the release of a single realization of $Y$. We also study these two notions of leakage with oblivious (or memoryless) guessing. We obtain closed-form expressions for all these leakage measures, with the exception of one. Specifically, we are able to obtain closed-form expression for maximal guesswork leakage for the binary erasure source only; deriving expressions for arbitrary sources appears challenging. Some of the consequences of our results are -- a connection between guesswork and differential privacy and a new operational interpretation to maximal $\alpha$-leakage in terms of guesswork.
\end{abstract}

\section{Introduction}\label{ection:intro}
Quantification of information leakage plays a crucial role in many applications, for example, in ensuring the security of sensitive data within communication systems, evaluating the efficiency of cryptographic protocols, in safeguarding sensitive information, and analyzing privacy-preserving models in federated learning among others. The fundamental goal in information leakage is to quantify how much information does data released to an adversary reveal about correlated sensitive data. This has been addressed by various works in the information theory literature~\cite{smith2009foundations,braun2009quantitative,du2012privacy,osia2019privacy,IssaWK20,LiaoSKC20,farokhi2021measuring,SaeidianCOS23,zarrabian2023lift}.

A prominent theme in the literature on quantifying information leakage involves the development of leakage measures with \emph{operational interpretation}. This approach ensures that the amount of information leaked is directly linked to specific security guarantees. The works~\cite{smith2009foundations,braun2009quantitative,IssaWK20,LiaoSKC20,SaeidianCOS23} use a \emph{guessing} framework to propose various operationally meaningful leakage measures with a focus on the multiplicative increase, upon the observation of a released random variable, in the probability of accurately guessing a sensitive random variable. In particular, Issa~\emph{et al.}~\cite{IssaWK20} introduce \emph{maximal leakage} as the logarithm of the multiplicative increase, upon observing $Y$, of the probability of correctly guessing a randomized function of $X$ in a single try, maximized over all such randomized functions.

In this paper, we study information leakage with emphasis on the guessing framework. However, rather than assessing the adversary's performance based on the probability of correctness in a single attempt, we allow the adversary to make \emph{any} number of guesses. We measure the performance through the minimum expected number of guesses required to accurately predict a random variable, which is termed as \emph{guesswork}~\cite{massey1994guessing,pliam1998disparity,arikan1996inequality}. In particular, we study information leakage through guesswork. 

We define \emph{maximal guesswork leakage} as the multiplicative decrease, upon observing $Y$, of the guesswork of a randomized function of $X$, maximized over all such randomized functions\footnote{This notion of leakage is, in part, inspired by an observation by Arikan~\cite{arikan1996inequality} which states that the asymptotic multiplicative decrease, upon observing $Y$, in the guesswork of $X$ can be interpreted as a \emph{complexity reduction} provided by the knowledge of $Y$ in guessing the value of $X$.}. We also study a pointwise form of the leakage, called \emph{pointwise maximal guesswork leakage} which captures the leakage due to the release of a single realization $y$ of $Y$ rather than the average outcome of $Y$. We also explore information leakage through guesswork in the context of \emph{oblivious} or \emph{memoryless guessing}~\cite{HanawalS12,boztas2012oblivious,HuleihelSM17,Slamatianetal19}, wherein an adversary cannot keep track of previous guesses. In particular, we study analogous leakages for oblivious guessing. It is worth noting that the works \cite{osia2019privacy} and \cite{farokhi2021measuring} also study information leakage based on number of guesses. In particular, the authors in \cite{osia2019privacy} study information leakage using guesswork for the scenario when an adversary is interested in guessing $X$ itself, instead of a possibly randomized function of $X$ as we do here. A non-stochastic setting of guessing is considered in \cite{farokhi2021measuring}.



Our main contributions are as follows:

\begin{itemize}
    \item We show that the pointwise maximal guesswork leakage is equal to the R\'{e}nyi divergence of order infinity between the \emph{a priori} distribution $P_X$ and the \emph{a posteriori} distribution $P_{X|Y=y}$ (Theorem~\ref{thm:pointwiseleakage}, and a generalization, Theorem~\ref{thm:pointwiseleakagewithh}, including $\rho$-th moments of guessing number). A consequence of this establishes a connection between guesswork and differential privacy (Corollary~\ref{Corollary:pwleakagewithDP}).
    \item We obtain a closed-form expression for maximal guesswork leakage for the binary erasure source~(Theorem~\ref{theorem:BES}). Deriving a closed-form expression for the leakage with an arbitrary distribution $P_{XY}$ appears challenging.
    \item We show that oblivious maximal $\rho$-guesswork leakage (Definition~\ref{defn:obmgl}) is proportional to the Arimoto channel capacity of order $\alpha=\frac{1}{1+\rho}$~\cite{arimoto1977information} (Theorem~\ref{thm:obmgl}). This provides a new operational interpretation to maximal $\alpha$-leakage~\cite{LiaoKS20} in terms of guesswork.
    \item Finally, we show that the pointwise oblivious maximal $\rho$-guesswork leakage (Definition~\ref{defn:pobmgl}) is equal to the R\'{e}nyi divergence of order infinity between $P_X$ and $P_{X|Y=y}$ (Theorem~\ref{thm:pwrhoguesleak}). It is interesting to note that this leakage does not depend on the value of $\rho$.
\end{itemize}

\section{Preliminaries}
\subsection{R\'{e}nyi Information Measures}
Most of the existing leakage measures and the new leakage measures we study in this paper can be expressed in terms of R\'{e}nyi information measures.
\begin{definition}[R\'{e}nyi divergence~\cite{renyi1961measures}]
The R\'{e}nyi divergence of order $\alpha\in(0,1)\cup(1,\infty)$ between two probability distributions $P_X$ and $Q_X$ on a finite alphabet $\mathcal{X}$ is defined as
\begin{align}\label{eqn:Renyidiv-def}
    D_\alpha(P_X||Q_X)=\frac{1}{\alpha-1}\log\left(\sum_{x\in\mathcal{X}}P_X(x)^\alpha Q_X(x)^{1-\alpha}\right).
\end{align}
It is defined by its continuous extension for $\alpha=1$ and $\alpha=\infty$, respectively, and is given by
\begin{align}
    D_1(P_X||Q_X)&=\sum_{x\in\mathcal{X}}P_X(x)\log\frac{P_X(x)}{Q_X(x)},\\
    D_\infty(P_X||Q_X)&=\max_{x\in\mathcal{X}}\log\frac{P_X(x)}{Q_X(x)}.
\end{align}
\end{definition}
\begin{definition}[Arimoto mutual information~\cite{arimoto1977information}]
    Given a joint distribution $P_{XY}$ on a finite alphabet $\mathcal{X}\times \mathcal{Y}$, the Arimoto mutual information of order $\alpha\in(0,1)\cup (1,\infty)$ is defined as
    \begin{align}
        I_\alpha^{\text{A}}(X;Y)=H_\alpha(X)-H_\alpha^{\text{A}}(X|Y),
    \end{align}
    where R\'{e}nyi entropy $H_\alpha(X)$~\cite{renyi1961measures} and Arimoto conditional entropy $H_{\alpha}^{\text{A}}(X|Y)$~\cite{arimoto1977information} are given by
    \begin{align}
        H_\alpha(X)=\frac{1}{1-\alpha}\log{\sum_{x\in\mathcal{X}}P_X(x)^\alpha}
    \end{align}
    and 
    \begin{align}
        H_\alpha^{\text{A}}(X|Y)=\frac{\alpha}{1-\alpha}\log{\!\sum_{y\in\emph{supp}(Y)}\!\!P_Y(y)\left(\sum_{x\in\mathcal{X}}P_{X|Y}(x|y)^\alpha\right)^\frac{1}{\alpha}},
    \end{align}
    respectively.
\end{definition}

\subsection{Guesswork}
Consider an adversary interested in guessing the realization of a random variable $X$ by asking questions of the form ``Is $X=x$?" until the answer is ``Yes".
\begin{definition}[Guesswork~\cite{massey1994guessing,pliam1998disparity}]
    A function $G:\mathcal{X}\rightarrow [1:|\mathcal{X}|]$ is called a guessing function for a random variable $X$ taking values in $\mathcal{X}$ if $G$ is one-to-one. Given a guessing function $G$, guessing number $G(x)$ is the number of guesses required to guess $x$, i.e., the time index of the question 'Is $X=x$?'. The guesswork is the minimum of the expected number of guesses required to guess $X$, i.e., $\min_{G}\mathbb{E}[G(X)]$, where the minimum is over all guessing functions $G$.
\end{definition}
Let $P_X$ be the probability distribution of $X$ taking values in $\mathcal{X}=\{x_1,x_2,\dots,x_n\}$  and suppose that $P_X(x_i)\geq P_X(x_{i+1})$, for $i\in[1:n-1]$, without loss of generality. The optimal guessing strategy is to guess in non-increasing order of probability values. So, we have
\begin{align}
     \min_{G}\mathbb{E}[G(X)]=\sum_{i=1}^n iP_X(x_i).
\end{align}
Arikan~\cite{arikan1996inequality} studied the $\rho$-th moments of guessing number for $\rho>0$, and obtained bounds on the same:
\begin{align}
     \min_{G}\mathbb{E}[G(X)^\rho]=\sum_{i=1}^n i^\rho P_X(x_i).
\end{align}
Salamatian~\emph{et~al.}~\cite{Slamatianetal19} considered oblivious guessing (also called memoryless guessing), wherein an adversary cannot keep track of the previous guesses. In particular, an adversary presents a sequence of independent and identically distributed (i.i.d.) guesses $\hat{X}_1^\infty:=(\hat{X}_1,\hat{X}_2,\dots)$ drawn from some distribution $P_{\hat{X}}$ to guess $X$. The number of guesses until a success is defined as the corresponding guessing number:
\begin{align}
    {G}(X,\hat{X}_1^\infty)=\inf\{k\geq 1:\hat{X}_k=X\}.
\end{align}
Analogous to the $\rho$-th moment of guessing number, Salamatian~\emph{et~al.}~\cite{Slamatianetal19} studied the following optimization problem:
\begin{align}\label{eqn:Medard-factorialfn}
    \inf_{P_{\hat{X}}}\mathbb{E}\left[V_\rho(X,\hat{X}_1^\infty)\right],
\end{align}
where $V_\rho(X,\hat{X}_1^\infty)=\binom{{G}(X,\hat{X}_1^\infty)+\rho-1}{\rho}$, $\rho>0$, and $\binom{x}{y}$ is the generalized binomial coefficient defined in terms of the gamma function $\Gamma(\cdot)$ as $\binom{x}{y}=\frac{\Gamma(x+1)}{\Gamma(y+1)\Gamma(x-y+1)}$.

\subsection{Maximal Leakage}
\begin{definition}[Maximal leakage\cite{IssaWK20}]\label{defn:maxL}
Given a joint distribution $P_{XY}$ on a finite alphabet $\mathcal{X}\times\mathcal{Y}$, the maximal leakage from $X$ to $Y$ is defined as 
\begin{align}
    \mathcal{L}(X\rightarrow Y)=\sup_{U:U-X-Y}\log{\frac{\sup_{P_{\hat{U}|Y}}\mathbb{E}[P_{\hat{U}|Y}(U|Y)]}{\sup_{P_{\hat{U}}}\mathbb{E}[P_{\hat{U}}(U)]}}.
\end{align} 
\end{definition}
From the above definition, maximal leakage is logarithm of the multiplicative increase, upon observing $Y$, of the probability of correctly guessing a randomized function of $X$, maximized over all such randomized functions.
Issa~\emph{et al.}~\cite[Theorem~1]{IssaWK20} showed that     
\begin{align}
    \mathcal{L}(X\rightarrow Y)=\log{\sum_{y\in\mathcal{Y}}\max_{x\in\text{supp}(X)}P_{Y|X}(y|x)}.
\end{align}

\section{Maximal Guesswork Leakage}\label{section:mgl}
Suppose an adversary is interested in guessing a randomized function $U$ of a hidden random variable $X$ by asking questions of the form ``Is $U$ equal to $u$?'' until the answer is ``Yes''.  The guesswork of $U$, i.e., $\min_{G}\mathbb{E}[G(U)]$, can be viewed as a cost incurred by adversary in guessing $U$. We define maximal guesswork leakage as follows.  
\begin{definition}[Maximal guesswork leakage]\label{definition:maxLviaGE}
    Let $P_{XY}$ be a joint distribution on a finite alphabet $\mathcal{X}\times\mathcal{Y}$. The maximal guesswork leakage from $X$ to $Y$ is defined by
    \begin{align}\label{eqn:maxLviaGE}
    \mathcal{L}^{G}(X\rightarrow Y)=\sup_{U:U-X-Y}\!\!\log{\frac{\min_{G}\mathbb{E}[G(U)]}{\min_{\{G_y: y\in\mathcal{Y}\}}{\mathbb{E}[G_Y(U)]}}},
    \end{align}
    where $\{G_y:y\in\mathcal{Y}\}$ is collection of guessing functions, one for each $y$, $\mathbb{E}[G_Y(U)]=\sum_{y\in\mathcal{Y}}P_Y(y)\mathbb{E}[G_y(U)|Y=y]$ and $U$ takes values in an arbitrary finite alphabet.
\end{definition}
\begin{remark}
    Maximal guesswork leakage in \eqref{eqn:maxLviaGE} is the multiplicative decrease, upon observing $Y$, of the guesswork of a randomized function of $X$, maximized over all such randomized functions. The notion of supremum over all randomized functions $U$ such that $U-X-Y$ forms a Markov chain in \eqref{eqn:maxLviaGE} is adapted from the framework of maximal leakage~\cite{IssaWK20}. 
\end{remark}
We next define a pointwise version of maximal guesswork leakage.
\begin{definition}[Pointwise maximal guesswork leakage]\label{definition:PWmaxLviaGE}
    Let $P_{XY}$ be a joint distribution on a finite alphabet $\mathcal{X}\times\mathcal{Y}$. The pointwise maximal guesswork leakage from $X$ to $y\in\emph{supp}(Y)$ is defined by
    \begin{align}\label{eqn:PWmaxLviaGE}
        \!\mathcal{L}^{G\emph{-pw}}(X\rightarrow y)=\! \sup_{U:U-X-Y}\log{\frac{\min_{G}\mathbb{E}[G(U)]}{\min_G{\mathbb{E}[G(U)|Y=y]}}},
    \end{align}
    where $U$ takes values in an arbitrary finite alphabet.
\end{definition}
    In view of the interpretation of guesswork as a cost incurred to a guessing adversary, maximal guesswork leakage in \eqref{definition:maxLviaGE} is related to maximal cost leakage studied by Issa~\emph{et~al.}~\cite[Section~VI-E]{IssaWK20}. In the same manner, pointwise maximal guesswork leakage is related to maximal realizable cost leakage~\cite[Section~VI-E]{IssaWK20}. We outline the distinction between the measures in \eqref{eqn:maxLviaGE} and \eqref{eqn:PWmaxLviaGE}, and that of Issa~\emph{et~al.}~\cite[Definitions~11 and 12]{IssaWK20} below.
    \begin{remark}\label{remark2}
    Issa~\emph{et al.}~\cite[Definitions~11]{IssaWK20} define maximal cost leakage as
    \begin{align}\label{eqn:maximalcostleakage}
        \mathcal{L}^{\text{c}}(X\rightarrow Y)=\sup_{\substack{U:U-X-Y\\ \hat{\mathcal{U}},d:\mathcal{U}\times\hat{\mathcal{U}}\rightarrow \mathbb{R}_+}}\log{\frac{\inf_{\hat{u}\in\hat{\mathcal{U}}}\mathbb{E}[d(U,\hat{u})]}{\inf_{\hat{u}(\cdot)}\mathbb{E}[d(U,\hat{u}(Y))]}}.
    \end{align}
    Note that the expression in \eqref{eqn:maximalcostleakage}, when $d(u,\hat{u})$ is viewed as the cost incurred in guessing $u$ as $\hat{u}$, concerns with the maximum reduction in cost that the adversary incurs. Consider, for $y\in\text{supp}(Y)$,
    \begin{align}\label{eqn:maximalcostleaakgeissa}
    \mathcal{L}^\text{pc}(X\rightarrow y)=\sup_{\substack{U:U-X-Y\\ \hat{\mathcal{U}},d:\mathcal{U}\times\hat{\mathcal{U}}\rightarrow \mathbb{R}_+}}\log{\frac{\inf_{\hat{u}\in\hat{\mathcal{U}}}\mathbb{E}[d(U,\hat{u})]}{\inf_{\hat{u}}\mathbb{E}[d(U,\hat{u})|Y=y]}}.
\end{align}
    
    Issa~\emph{et al.}~\cite[Definition~12]{IssaWK20} define $\max_{y\in\text{supp}(Y)}\mathcal{L}^\text{c}(X\rightarrow y)$ as maximal realizable cost leakage. The guessing number $G(u)$ is a special case of the cost function $d(u,\hat{u})$ (considered in \eqref{eqn:maximalcostleakage} and \eqref{eqn:maximalcostleaakgeissa}):  define $\hat{\mathcal{U}}$ to be the set of all permutations of $\mathcal{U}$ and
    \begin{align}\label{eqn:gtoc}
    d(u,\hat{u})=\sum_{i=1}^{|\mathcal{U}|}i\mathbbm{1}\{u=\hat{u}_i\},
    \end{align}
    where each permutation $\hat{u}=(\hat{u}_1,\hat{u}_2,\dots,\hat{u}_{|\mathcal{U}|})$ inherently determines a guessing function $G$. The expressions in \eqref{eqn:maximalcostleakage} and \eqref{eqn:maximalcostleaakgeissa} consider the worst-case scenario over all cost functions by taking supremum over $d$ and $\hat{\mathcal{U}}$. Thus, if we fix particular choice of $d$ as in \eqref{eqn:gtoc},  we get \eqref{eqn:maxLviaGE} and \eqref{eqn:PWmaxLviaGE} from \eqref{eqn:maximalcostleakage} and \eqref{eqn:maximalcostleaakgeissa}, respectively. So, we have the following upper bounds on \eqref{eqn:maxLviaGE} and \eqref{eqn:PWmaxLviaGE}.
    \begin{align}
         \mathcal{L}^{G}(X\rightarrow Y)&\leq \mathcal{L}^\text{c}(X\rightarrow Y),\label{eqn:boundonavg}\\
          \mathcal{L}^{G\text{-pw}}(X\rightarrow y)&\leq  \mathcal{L}^\text{pc}(X\rightarrow y), \ y\in\text{supp}(Y)\label{eqn:boundonpoint}.
    \end{align}
    Issa~\emph{et~al.}~\cite[Theroems~15 and 16]{IssaWK20} obtained closed-form expressions for \eqref{eqn:maximalcostleakage} and \eqref{eqn:maximalcostleaakgeissa} by constructing a counterintuitive\footnote{The cost function $d(u,\hat{u})=\frac{1}{P_U(u)}\mathbbm{1}\{u=\hat{u}\}$ is counterintuitive in that it is maximum when the adversary's guess is correct and it is minimum when the guess is wrong.} cost function $d(u,\hat{u})=\frac{1}{P_U(u)}\mathbbm{1}\{u=\hat{u}\}$. Substituting those expressions into \eqref{eqn:boundonavg} and \eqref{eqn:boundonpoint} gives  
     \begin{align}
          \mathcal{L}^{G}(X\rightarrow Y)&\leq -\log{\sum_{y\in\mathcal{Y}}\min_{x\in\mathcal{X}:P_X(x)>0}P_{Y|X}(y|x)},\label{eqn:upperboundonmaxLGE}\\
           \mathcal{L}^{G\text{-pw}}(X\rightarrow y)&\leq D_{\infty}(P_X\|P_{X|Y=y}).\label{eqn:upperboundonpwmaxLGE}
    \end{align}
    Note that {\cite[Theorem~16]{IssaWK20}} actually considers $\max_{y\in\text{supp}(Y)}\mathcal{L}^\text{pc}(X\rightarrow y)$, however, the expression for $\mathcal{L}^\text{pc}(X\rightarrow y)$ can be inferred from \cite[Proof of Theroem~16]{IssaWK20}. 
\end{remark}
Interestingly, we show that the bound in \eqref{eqn:upperboundonpwmaxLGE} is tight while the bound in \eqref{eqn:upperboundonmaxLGE} is not tight, in general. The following theorem shows that the bound in \eqref{eqn:upperboundonpwmaxLGE} is tight. 
\begin{theorem}[Pointwise maximal guesswork leakage ]\label{thm:pointwiseleakage}
    For any joint probability distribution $P_{XY}$ on a finite alphabet $\mathcal{X}\times\mathcal{Y}$, the pointwise maximal guesswork leakage from $X$ to $y\in\emph{supp}(Y)$ is given by
    \begin{align}
        \mathcal{L}^{G\emph{-pw}}(X\rightarrow y)=D_\infty(P_X\| P_{X|Y=y}).
    \end{align}
\end{theorem}
We in fact prove a more general version of the above theorem which considers a function of the guessing number instead. Let  $\mathcal{L}^{h(G)\text{-pw}}(X\rightarrow y)$ denote the corresponding leakage where $h:\mathbb{N}\rightarrow \mathbb{R}_+$ is a function of the guessing number.
\begin{align}\label{eqn:pwleakagewithh}
 \mathcal{L}^{h(G)\text{-pw}}(X\rightarrow y)= \sup_{U:U-X-Y}\log{\frac{\min_{G}\mathbb{E}\left[h(G(U))\right]}{\min_G{\mathbb{E}\left[h(G(U))|Y=y\right]}}}.
 \end{align}
 Notice that $h(G(u))$ can also be viewed as a cost function defined by 
 \begin{align}\label{eqn:guessascost}
 d(u,\hat{u})=\sum_{i=1}^{|\mathcal{U}|}h(i)\mathbbm{1}\{u=\hat{u}_i\},
 \end{align}
 where $\hat{\mathcal{U}}$ is the set of all permutations of $\mathcal{U}$ and each permutation $\hat{u}=(\hat{u}_1,\hat{u}_2,\dots,\hat{u}_{|\mathcal{U}|})$ inherently determines a guessing function $G$.
\begin{theorem}\label{thm:pointwiseleakagewithh}
     Let $P_{XY}$ be a joint distribution on a finite alphabet $\mathcal{X}\times\mathcal{Y}$ and $h:\mathbb{N}\rightarrow \mathbb{R}_+$ be a non-decreasing function such that $h(n)\rightarrow \infty$ as $n\rightarrow \infty$. Then we have, for $y\in\emph{supp}(Y)$,
     \begin{align}\label{eqn:pointwiseleakagewithh1}
         \mathcal{L}^{h(G)\emph{-pw}}(X\rightarrow y)=D_\infty({P_X\| P_{X|Y=y}}).
     \end{align}
\end{theorem}
\begin{remark}
    In relation to \eqref{eqn:maximalcostleaakgeissa}, Theorem~\ref{thm:pointwiseleakagewithh} shows the optimality of the cost function corresponding to guessing number. Specifically, recall that $\mathcal{L}^\text{pc}(X\rightarrow y)$ in \eqref{eqn:maximalcostleaakgeissa} considers supremum over all cost functions $d$ (and also $\hat{\mathcal{U}}$) and its closed-form expression is given by the RHS of \eqref{eqn:upperboundonpwmaxLGE}. Theorem~\ref{thm:pointwiseleakagewithh} shows that this supremum in \eqref{eqn:maximalcostleaakgeissa} is achieved by a class of cost functions in \eqref{eqn:guessascost} corresponding to guessing number that are more operationally motivated than the counterintuitive cost function $d(u,\hat{u})=\frac{1}{P_U(u)}\mathbbm{1}\{u=\hat{u}\}$ (see Footnote~2) which is shown to achieve the supremum in \cite[Theorem~16]{IssaWK20}.
\end{remark}
\begin{remark}
    An important example of a function $h$ satisfying the conditions in Theorem~\ref{thm:pointwiseleakagewithh} is $h(n)=n^\rho$, $\rho\in(0,\infty)$, which corresponds to moments of guessing number first studied by Arikan~\cite{arikan1996inequality}. Also, note that Theorem~\ref{thm:pointwiseleakagewithh} with $h(n)=n^\rho$ recovers Theorem~\ref{thm:pointwiseleakage} when $\rho=1$. Some other examples of $h(n)$ that satisfy the conditions in Theorem~\ref{thm:pointwiseleakagewithh} are 
    \begin{align}
        \log{n}, \  \frac{\mathsf{e}^n}{n+1},\ a^n\ \text{where}\ a>1.
    \end{align}
\end{remark}
\begin{remark}
For a $y\in\text{supp}(Y)$, if there exists an $x^*\in\mathcal{X}$ such that $P_X(x^*)>0$ and $P_{X|Y}(x^*|y^*)=0$, then $D_\infty(P_X\|P_{X|Y=y})=\infty$. Theorem~\ref{thm:pointwiseleakagewithh} implies that $\mathcal{L}^{h(G)\text{-pw}}(X\rightarrow y)$ is also equal to infinity for such distributions. In fact, such distributions completely characterize the set of all distributions for which $\mathcal{L}^{h(G)\text{-pw}}(X\rightarrow y)=\infty$.
\end{remark}
\begin{proof}[Proof sketch of Theorem~\ref{thm:pointwiseleakagewithh}]
    The upper bound follows analogous to the discussion corresponding to \eqref{eqn:upperboundonpwmaxLGE} in Remark~\ref{remark2}. We use the `shattering' conditional distribution $P_{U|X}$~\cite[Proof of Theorem~1]{IssaWK20},\cite[Proof of Theorem~5]{LiaoSKC20} to prove the lower bound. A detailed proof is given in 
    \if \extended 1%
    Appendix~\ref{Proof:thm:pointwiseleakagewithh}.
    \fi
    \if \extended 0%
    an extended version~\cite[Appendix~A]{kurriMP24}.
    \fi
\end{proof}
\begin{corollary}[Guesswork and differential privacy]\label{Corollary:pwleakagewithDP}
For any conditional distribution $P_{Y|X}$ with $X$ and $Y$ taking values in finite alphabets $\mathcal{X}$ and $\mathcal{Y}$,  
\begin{align}
    \max_{P_X}\max_{y\in\emph{supp}(Y)}  \mathcal{L}^{G\emph{-pw}}(X\rightarrow y)=\max_{\substack{x,x^\prime\in\mathcal{X}\\ y\in\mathcal{Y}}}\log{\frac{P_{Y|X}(y|x)}{P_{Y|X}(y|x^\prime)}},
\end{align}
where the expression in RHS corresponds to the leakage measure in terms of local differential privacy for $P_{Y|X}$~\cite{kasiviswanathan2011can}.
\end{corollary}
 The proof of Corollary~\ref{Corollary:pwleakagewithDP} follows from Theorem~\ref{thm:pointwiseleakagewithh} and by using \cite[Corollary~7]{IssaWK20} which connects the R\'{e}nyi divergence of order infinity in RHS of \eqref{eqn:pointwiseleakagewithh1} with local differential privacy. Corollary~\ref{Corollary:pwleakagewithDP} provides an operational interpretation to local differential privacy in terms of information leakage using guesswork. 

We show that the upper bound in \eqref{eqn:upperboundonmaxLGE} is not tight in general, through an example. In particular, we characterize the maximal guesswork leakage for a binary erasure source.
\begin{theorem}[Maximal guesswork leakage for the binary erasure source]\label{theorem:BES}
    Consider a binary erasure source on $\{0,1\}\times \{0,\emph{e},1\}$ with joint distribution given by
    \begin{subequations}
    \begin{align}
        P_{XY}(i,i)&=\frac{1-p}{2}, i\in\{0,1\}\label{eqn:erasuresource1}.\\
        P_{XY}(i,\emph{e})&=\frac{p}{2}, i\in\{0,1\}\label{eqn:erasuresource2},
    \end{align}
     \end{subequations}
    where $p\in[0,1)$. Then maximal guesswork leakage is given by
    \begin{align}\label{MLGE-BES}
         \mathcal{L}^{G}(X\rightarrow Y)=\log{\frac{2}{1+p}}.
    \end{align}
\end{theorem}
\begin{remark}
    Theorem~\ref{theorem:BES} gives a closed-form expression for maximal guesswork leakage for a binary erasure source. Note that for $p\in[0,1)$, $\log{\frac{2}{1+p}}<\log{\frac{1}{p}}$ which is the upper bound in \eqref{eqn:upperboundonmaxLGE} for the binary erasure source. Obtaining a closed-form expression for $\mathcal{L}^{G}(X\rightarrow Y)$ for an arbitrary $P_{XY}$ appears challenging. 
\end{remark}
\begin{proof}[Proof sketch of Theorem~\ref{theorem:BES}]
    We establish a more stringent upper bound through our proof, surpassing the bound in \eqref{eqn:upperboundonmaxLGE} that would result from employing maximal cost leakage. Let $\gamma(P)$ denote the guesswork of a random variable with probability distribution $P$ over $\mathcal{U}$. Using $A$ and $B$ to denote $P_{U|X=0}$ and $P_{U|X=1}$ respectively, it follows that,
\begin{align}
 &\mathcal{L}^{G}(X\rightarrow Y)\nonumber\\
       &=\sup_{P_{U|X}}\frac{1}{(1-p)\left(\frac{\frac{1}{2}\gamma(A)+\frac{1}{2}\gamma(B)}{\gamma(\frac{A+B}{2})}\right)+p}\label{eqn:erasureopt-1main}\\
       &=\sup_{P_U}\sup_{A,B\in\mathcal{P}:A+B=2P_U}\frac{1}{(1-p)\left(\frac{\frac{1}{2}\gamma(A)+\frac{1}{2}\gamma(B)}{\gamma(\frac{A+B}{2})}\right)+p}\label{eqn:erasureoptmain}\\
    &=\sup_{P_U}\frac{1}{(1-p)\left(\inf\limits_{A,B\in\mathcal{P}:A+B=2P_U}\frac{\frac{1}{2}\gamma(A)+\frac{1}{2}\gamma(B)}{\gamma(P_U)}\right)+p}\label{eqn:neweqn}
    \end{align}
    
   In the complete proof of Theorem~\ref{theorem:BES} in 
   \if \extended 0%
   the extended version~\cite{kurriMP24},
   \fi
   \if \extended 1%
   Appendix~\ref{proof:MLGE-BES},
   \fi
    we prove the following claim, which is the key ingredient of the proof of the theorem. 
   
    \begin{restatable}{claim}{restatedclaim}\label{claim}
Fix a $P_U$ over $\mathcal{U}=\{u_1,u_2,\dots,u_n\}$, where $n$ is an even number, such that $P_U(u_1)\geq P_U(u_2)\geq \dots \geq P_U(u_n)$. Then we have
\begin{align}\label{eqn:erasureproof5}
    \inf_{A,B\in\mathcal{P}:A+B=2P_U}\frac{1}{2}\gamma(A)+\frac{1}{2}\gamma(B)\nonumber\\
   \geq \sum_{i=1}^{\frac{n}{2}}i(P_U(u_{2i})+P_U(u_{2i-1})), 
\end{align}
   where $\mathcal{P}$ is the set of all probability distributions on $\mathcal{U}$.
\end{restatable}
The intuition for Claim~\ref{claim} is that if we relax the condition in the optimization problem that $A$ and $B$ need to be probability distributions and allow them to be non-negative real vectors of length $|\mathcal{U}|$ by considering an appropriate extension of the definition of $\gamma(\cdot)$, then the optimal $A^*$ and $B^*$ for the corresponding optimization problem are given by $A^*(u_{2i})=2P_U(u_{2i})$, $B^*(u_{2i-1})=2P_U(u_{2i-1})$, for $i\in[1:\frac{n}{2}]$. 
 
 Now continuing \eqref{eqn:neweqn}, we get


\begin{align}
&\mathcal{L}^{G}(X\rightarrow Y)\nonumber\\
    &\leq\sup_{P_U}\frac{1}{(1-p)\left(\frac{\sum_{i=1}^{\frac{n}{2}}i(P_U(u_{2i})+P_U(u_{2i-1}))}{\sum_{i=1}^niP_U(u_i)}\right)+p}\label{eqn:erasurethmupperbnd1main}\\
    &\leq \frac{1}{(1-p)\frac{1}{2}+p}\label{eqn:erasureproof4main}\\
    &=\frac{2}{1+p},
\end{align}
where the supremum in \eqref{eqn:erasureopt-1main} is decomposed into two suprema  in \eqref{eqn:erasureoptmain} using the fact that $P_U(u)=P_{U|X}(u|0)P_X(0)+P_{U|X}(u|1)P_X(1)=\frac{A(u)+B(u)}{2}$; \eqref{eqn:erasurethmupperbnd1main} follows from Claim~\ref{claim}; \eqref{eqn:erasureproof4main} and holds because
\begin{align}
\frac{\sum_{i=1}^{\frac{n}{2}}i(P_U(u_{2i})+P_U(u_{2i-1}))}{\sum_{i=1}^niP_U(u_i)}\geq \frac{1}{2}.
\end{align}
We use the `shattering' conditional distribution $P_{U|X}$~\cite[Proof of Theorem~1]{IssaWK20},\cite[Proof of Theorem~5]{LiaoSKC20} to prove the lower bound. The complete proof is given in %
\if \extended 1%
Appendix~\ref{proof:MLGE-BES}.
\fi
\if \extended 0%
\cite[Appendix~B]{kurriMP24}
\fi
\end{proof}

\section{Oblivious Maximal $\rho$-Guesswork Leakage}
In Section~\ref{section:mgl}, we considered the setup where there are no constraints on the memory of adversary, i.e., with each new attempt, the adversary is aware of their past guesses and avoids repeating any previously incorrect ones. Here, we consider a memoryless adversary which cannot keep track of previous guesses. 
\begin{definition}[Oblivious maximal $\rho$-guesswork leakage]\label{defn:obmgl}
     Let $P_{XY}$ be a joint distribution on a finite alphabet $\mathcal{X}\times\mathcal{Y}$. The oblivious maximal $\rho$-guesswork leakage from $X$ to $Y$ is defined by
     \begin{align}
         &\mathcal{L}_\rho^{\emph{oblv}-G}(X\rightarrow Y)\nonumber\\
         &=\sup_{U:U-X-Y}\log{\frac{\inf_{P_{\hat{U}}}\mathbb{E}[V_\rho(U,\hat{U}_1^\infty)]}{\inf_{P_{\hat{U}|Y}}\sum_{y\in\mathcal{Y}}P_Y(y)\mathbb{E}[V_\rho(U,\hat{U}_1^\infty)|Y=y]}},
     \end{align}
     where $V_\rho(U,\hat{U}_1^\infty)$ is as defined in \eqref{eqn:Medard-factorialfn}.
\end{definition}
\begin{definition}[Pointwise oblivious maximal $\rho$-guesswork leakage]\label{defn:pobmgl}
     Let $P_{XY}$ be a joint distribution on a finite alphabet $\mathcal{X}\times\mathcal{Y}$. The pointwise oblivious maximal $\rho$-guesswork leakage from $X$ to $y$, for $y\in\emph{supp}(Y)$, is defined by
     \begin{align}
         &\mathcal{L}_\rho^{\emph{pw-oblv}-G}(X\rightarrow y)\nonumber\\
         &=\sup_{U:U-X-Y}\log{\frac{\inf_{P_{\hat{U}}}\mathbb{E}[V_\rho(U,\hat{U}_1^\infty)]}{\inf_{P_{\hat{U}|Y=y}}\mathbb{E}[V_\rho(U,\hat{U}_1^\infty)|Y=y]}},
     \end{align}
     where $V_\rho(U,\hat{U}_1^\infty)$ is as defined in \eqref{eqn:Medard-factorialfn}.
\end{definition}
\begin{theorem}[Oblivious maximal $\rho$-guesswork leakage]\label{thm:obmgl}
For any joint probability distribution $P_{XY}$ on a finite alphabet $\mathcal{X}\times\mathcal{Y}$, the oblivious maximal $\rho$-leakage from $X$ to $Y$, for $\rho>0$, is given by
\begin{align}\label{eqn:obmglthm}
    \mathcal{L}_\rho^{\emph{oblv}-G}(X\rightarrow Y)=\rho\sup_{P_{\tilde{X}}\ll P_{X}}I_{\frac{\rho}{1+\rho}}^{\text{A}}(\tilde{X};Y).
\end{align}
\end{theorem}
\begin{remark}
    We note that the right-hand-side of \eqref{eqn:obmglthm} is proportional to maximal $\alpha$-leakage~\cite{LiaoSKC20}, a generalization of maximal leakage, with $\alpha=\frac{\rho}{1+\rho}$. This provides a new operational interpretation to maximal $\alpha$-leakage in terms of guesswork.
\end{remark}
\begin{proof}[Proof sketch of Theorem~\ref{thm:obmgl}]
    Using \cite[Lemma~2 and (28)]{Slamatianetal19}, we show that 
    \begin{align}
        \mathcal{L}_\rho^{\text{oblv}-G}(X\rightarrow Y)=\rho \sup_{U:U-X-Y}I^{\text{A}}_{\frac{1}{1+\rho}}(U;Y).
    \end{align}
    Invoking \cite[Theorem~5]{LiaoSKC20} which states that
    \begin{align}
        \sup_{U:U-X-Y}I^{\text{A}}_{\frac{1}{1+\rho}}(U;Y)=\sup_{P_{\tilde{X}}\ll P_{X}}I_{\frac{\rho}{1+\rho}}^{\text{A}}(\tilde{X};Y)
    \end{align}
    completes the proof. A detailed proof is given in %
    \if \extended 1%
    Appendix~\ref{proof:obgml}.
    \fi
    \if \extended 0%
    \cite[Appendix~C]{kurriMP24}.
    \fi
\end{proof}
\begin{theorem}[Pointwise oblivious maximal $\rho$-guesswork leakage]\label{thm:pwrhoguesleak}
For any joint probability distribution $P_{XY}$ on a finite alphabet $\mathcal{X}\times\mathcal{Y}$, the oblivious maximal $\rho$-leakage from $X$ to $y$, for $\rho>0$, is given by
\begin{align}
    \mathcal{L}_\rho^{\emph{pw-oblv}-G}(X\rightarrow y)=D_\infty(P_{X}\| P_{X|Y=y}).
\end{align}
\end{theorem}
\begin{proof}[Proof sketch of Theorem~\ref{thm:pwrhoguesleak}]
    Using \cite[Lemma~2 and (28)]{Slamatianetal19}, we show that 
    \begin{align}
        \mathcal{L}_\rho^{\text{pw-oblv}-G}&(X\rightarrow y)\nonumber\\
        &=\sup_{U:U-X-Y}\log\frac{\left(\sum_{u\in\mathcal{U}}P_U(u)^\alpha\right)^\frac{1}{\alpha}}{\left(\sum_{u\in\mathcal{U}}P_{U|Y}(u|y)^\alpha\right)^\frac{1}{\alpha}},
    \end{align}
    for $\alpha=\frac{1}{1+\rho}\leq 1$. We then show that
    \begin{align}\label{eqn:thm5assert1}
      \sup_{U:U-X-Y}\log\frac{\left(\sum_{u\in\mathcal{U}}P_U(u)^\alpha\right)^\frac{1}{\alpha}}{\left(\sum_{u\in\mathcal{U}}P_{U|Y}(u|y)^\alpha\right)^\frac{1}{\alpha}}=\max_{x\in\mathcal{X}}\frac{P_X(x)}{P_{X|Y}(x|y)}.
    \end{align}
    The proof of the upper bound in \eqref{eqn:thm5assert1} is in spirit along the lines of \cite[Proposition~5]{IssaWK20}. We use the `shattering' conditional distribution $P_{U|X}$~\cite[Proof of Theorem~1]{IssaWK20},\cite[Proof of Theorem~5]{LiaoSKC20} to prove the lower bound. A detailed proof is given in %
    \if \extended 1%
    Appendix~\ref{proof:pwrhoguesleak}.
    \fi
    \if \extended 0%
    \cite[Appendix~D]{kurriMP24}.
    \fi
\end{proof}

\section{Acknowledgment}
Gowtham R. Kurri would like to thank Oliver Kosut and Lalitha Sankar for helpful discussions that contributed to the proof of \eqref{eqn:thm5assert1}.

\IEEEtriggeratref{11}

\bibliographystyle{IEEEtran}
\bibliography{Bibliography}

\begin{thebibliography}{10}
\providecommand{\url}[1]{#1}
\csname url@samestyle\endcsname
\providecommand{\newblock}{\relax}
\providecommand{\bibinfo}[2]{#2}
\providecommand{\BIBentrySTDinterwordspacing}{\spaceskip=0pt\relax}
\providecommand{\BIBentryALTinterwordstretchfactor}{4}
\providecommand{\BIBentryALTinterwordspacing}{\spaceskip=\fontdimen2\font plus
\BIBentryALTinterwordstretchfactor\fontdimen3\font minus
  \fontdimen4\font\relax}
\providecommand{\BIBforeignlanguage}[2]{{%
\expandafter\ifx\csname l@#1\endcsname\relax
\typeout{** WARNING: IEEEtran.bst: No hyphenation pattern has been}%
\typeout{** loaded for the language `#1'. Using the pattern for}%
\typeout{** the default language instead.}%
\else
\language=\csname l@#1\endcsname
\fi
#2}}
\providecommand{\BIBdecl}{\relax}
\BIBdecl

\bibitem{smith2009foundations}
G.~Smith, ``On the foundations of quantitative information flow,'' in
  \emph{International Conference on Foundations of Software Science and
  Computational Structures}, 2009, pp. 288--302.

\bibitem{braun2009quantitative}
C.~Braun, K.~Chatzikokolakis, and C.~Palamidessi, ``Quantitative notions of
  leakage for one-try attacks,'' \emph{Electronic Notes in Theoretical Computer
  Science}, vol. 249, pp. 75--91, 2009.

\bibitem{du2012privacy}
F.~du~Pin~Calmon and N.~Fawaz, ``Privacy against statistical inference,'' in
  \emph{50th annual Allerton conference on communication, control, and
  computing (Allerton)}.\hskip 1em plus 0.5em minus 0.4em\relax IEEE, 2012, pp.
  1401--1408.

\bibitem{osia2019privacy}
S.~A. Osia, B.~Rassouli, H.~Haddadi, H.~R. Rabiee, and D.~G{\"u}nd{\"u}z,
  ``Privacy against brute-force inference attacks,'' in \emph{2019 IEEE
  International Symposium on Information Theory (ISIT)}.\hskip 1em plus 0.5em
  minus 0.4em\relax IEEE, 2019, pp. 637--641.

\bibitem{IssaWK20}
I.~Issa, A.~B. Wagner, and S.~Kamath, ``An operational approach to information
  leakage,'' \emph{IEEE Transactions on Information Theory}, vol.~66, no.~3,
  pp. 1625--1657, 2020.

\bibitem{LiaoSKC20}
J.~Liao, L.~Sankar, O.~Kosut, and F.~P. Calmon, ``Maximal $\alpha$-leakage and
  its properties,'' in \emph{IEEE Conference on Communications and Network
  Security}, 2020, pp. 1--6.

\bibitem{farokhi2021measuring}
F.~Farokhi and N.~Ding, ``Measuring information leakage in non-stochastic
  brute-force guessing,'' in \emph{2020 IEEE Information Theory Workshop
  (ITW)}.\hskip 1em plus 0.5em minus 0.4em\relax IEEE, 2021, pp. 1--5.

\bibitem{SaeidianCOS23}
S.~Saeidian, G.~Cervia, T.~J. Oechtering, and M.~Skoglund, ``Pointwise maximal
  leakage,'' \emph{IEEE Transactions on Information Theory}, vol.~69, no.~12,
  pp. 8054--8080, 2023.

\bibitem{zarrabian2023lift}
M.~A. Zarrabian, N.~Ding, and P.~Sadeghi, ``On the lift, related privacy
  measures, and applications to privacy--utility trade-offs,'' \emph{Entropy},
  vol.~25, no.~4, p. 679, 2023.

\bibitem{massey1994guessing}
J.~L. Massey, ``Guessing and entropy,'' in \emph{Proceedings of 1994 IEEE
  International Symposium on Information Theory}, 1994, p. 204.

\bibitem{pliam1998disparity}
J.~O. Pliam, ``The disparity between work and entropy in cryptology.''
  \emph{IACR Cryptol. ePrint Arch.}, vol. 1998, p.~24, 1998.

\bibitem{arikan1996inequality}
E.~Arikan, ``An inequality on guessing and its application to sequential
  decoding,'' \emph{IEEE Transactions on Information Theory}, vol.~42, no.~1,
  pp. 99--105, 1996.

\bibitem{HanawalS12}
M.~J. {Hanawal} and R.~{Sundaresan}, ``Randomized attacks on passwords,'' in
  \emph{DRDO-IISc Programme on Advanced Research in Mathematical Engineering},
  2010.

\bibitem{boztas2012oblivious}
S.~Boztas, ``Oblivious distributed guessing,'' in \emph{IEEE International
  Symposium on Information Theory Proceedings}, 2012, pp. 2161--2165.

\bibitem{HuleihelSM17}
W.~Huleihel, S.~Salamatian, and M.~Médard, ``Guessing with limited memory,''
  in \emph{IEEE International Symposium on Information Theory (ISIT)}, 2017,
  pp. 2253--2257.

\bibitem{Slamatianetal19}
S.~Salamatian, W.~Huleihel, A.~Beirami, A.~Cohen, and M.~Médard, ``Why botnets
  work: Distributed brute-force attacks need no synchronization,'' \emph{IEEE
  Transactions on Information Forensics and Security}, vol.~14, no.~9, pp.
  2288--2299, 2019.

\bibitem{arimoto1977information}
S.~Arimoto, ``Information measures and capacity of order $\alpha$ for discrete
  memoryless channels,'' \emph{Topics in information theory}, 1977.

\bibitem{LiaoKS20}
J.~{Liao}, O.~{Kosut}, L.~{Sankar}, and F.~P. {Calmon}, ``Tunable measures for
  information leakage and applications to privacy-utility tradeoffs,''
  \emph{IEEE Transactions on Information Theory}, vol.~65, no.~12, pp.
  8043--8066, 2019.

\bibitem{renyi1961measures}
A.~R{\'e}nyi, ``On measures of entropy and information,'' in \emph{Proceedings
  of the Fourth Berkeley Symposium on Mathematical Statistics and Probability,
  Volume 1: Contributions to the Theory of Statistics}, vol.~4, 1961, pp.
  547--562.

\bibitem{kasiviswanathan2011can}
S.~P. Kasiviswanathan, H.~K. Lee, K.~Nissim, S.~Raskhodnikova, and A.~Smith,
  ``What can we learn privately?'' \emph{SIAM Journal on Computing}, vol.~40,
  no.~3, pp. 793--826, 2011.

\end{thebibliography}

\clearpage

\if \extended 1%
\appendices

\section{Proof of Theorem~\ref{thm:pointwiseleakagewithh}}\label{Proof:thm:pointwiseleakagewithh}
The proof of the upper bound follows from \cite[Theorem~16]{IssaWK20} by noting that $h(G(u))$ can be seen as a special case of $d(u,\hat{u})$ as mentioned in the paragraph before Theorem~\ref{thm:pointwiseleakagewithh}. In particular, this follows by defining $\hat{\mathcal{U}}$ as the set of all permutations of $\mathcal{U}$ and 
\begin{align}
    d(u,\hat{u})=\sum_{i=1}^{|\mathcal{U}|}h(i)\mathbbm{1}\{u=\hat{u}_i\},
\end{align}
noting that
\begin{align}
    \mathbb{E}[d(u,\hat{u})]=\sum_{i=1}^{|\mathcal{U}|}h(i)P_U(\hat{u}_i)=\mathbb{E}[h(G(U))],
\end{align}
for a guessing strategy inherently determined by permutation $\hat{u}$.

We prove the lower bound now. For this, we use the `shattering' conditional distribution $P_{U|X}$~\cite[Proof of Theorem~1]{IssaWK20},\cite[Proof of Theorem~5]{LiaoSKC20}. Let $\mathcal{U}=\cup_{x\in\mathcal{X}}\mathcal{U}_x$ (a disjoint union) and $|\mathcal{U}_x|=m_x$, for $x\in\mathcal{X}$. Define
\begin{align}
P_{U|X}(u|x)=\begin{cases}\frac{1}{m_x}, & u\in\mathcal{U}_x\\
0, & \text{otherwise}.
\end{cases}
\end{align}
Fix an $x^*\in\argmax_{x\in\mathcal{X}}\frac{P_X(x)}{P_{X|Y}(x|y)}$ and let $m_{x}=1$, for $x\neq x^*$. This gives,
\begin{align}
    P_U(u)=\begin{cases}
        P_X(x), & u\in\mathcal{U}_x, x\neq x^*\\
        \frac{P_X(x^*)}{m_{x^*}}, & u\in\mathcal{U}_{x^*}.
    \end{cases}
\end{align}
We denote the optimal guessing strategies in both the numerator and the denominator in \eqref{eqn:pwleakagewithh} by $G^*(\cdot)$ and $G_y^*(\cdot)$, respectively. Then, with sufficiently large $m_{x^*}$, we have
\begin{align}\label{eqn:proof:thm:pointwiseleakagewithhnume}
    \mathbb{E}\left[h(G^*(U))\right]=\sum_{i=1}^{|\mathcal{X}|-1}h(i)P_X(\tilde{x}_i)+\frac{P_X(x^*)}{m_{x^*}}\sum_{i=|\mathcal{X}|}^{|\mathcal{X}|-1+m_{x^*}}h(i),
\end{align}
where $(\tilde{x}_i)_{i\in[1:|\mathcal{X}|-1]}$ is a sequence in non-decreasing order of probabilities $P_X(x)$, $x\in\mathcal{X}\setminus\{x^*\}$. Similarly, for sufficiently large $m_{x^*}$,
\begin{align}\label{eqn:proof:thm:pointwiseleakagewithhdenome}
    \mathbb{E}[h(G_y^*(U))]=\!\sum_{i=1}^{|\mathcal{X}|-1}h(i)P_{X|Y}(\bar{x}_i|y)+P_{X|Y}(x^*|y)\tau,
\end{align}
where $(\bar{x}_i)_{i\in[1:|\mathcal{X}|-1]}$ is a sequence in non-decreasing order of probabilities $P_{X|Y}(x|y)$, $x\in\mathcal{X}\setminus\{x^*\}$ and $\tau=\frac{1}{m_{x^*}}\sum_{i=|\mathcal{X}|}^{|\mathcal{X}|-1+m_{x^*}}h(i)$. We show that $\tau\rightarrow \infty$ as $m_{x^*}\rightarrow \infty$. 
Consider
\begin{align}
    \sum_{i=1}^n\frac{h(i)}{n}&=\sum_{i=1}^k\frac{h(i)}{n}+\sum_{i=k+1}^n\frac{h(i)}{n}\\
    &\geq \sum_{i=1}^k\frac{h(i)}{n}+\frac{(n-k)h(k)}{n}\label{eqn:proof:thm:pointwiseleakagewithh1}\\
    &= h(k)+\epsilon_n,\label{eqn:proof:thm:pointwiseleakagewithh2}
\end{align}
where \eqref{eqn:proof:thm:pointwiseleakagewithh1} follows since $h(n)$ is a non-decreasing function and \eqref{eqn:proof:thm:pointwiseleakagewithh2} follows as $n\rightarrow \infty$ with $\epsilon_n=\sum_{i=1}^k\frac{h(i)}{n}\rightarrow 0$, for every $k\in\mathbb{N}$. Now since $h(n)\rightarrow \infty$ as $n\rightarrow \infty$, we have $\frac{1}{n}\sum_{i=1}^ nh(i)\rightarrow \infty$ as $n\rightarrow\infty$. This gives that $\tau=\frac{1}{m_{x^*}}\sum_{i=|\mathcal{X}|}^{|\mathcal{X}|-1+m_{x^*}}h(i)\rightarrow \infty$ as $m_{x^*}\rightarrow \infty$. Using \eqref{eqn:proof:thm:pointwiseleakagewithhnume} and \eqref{eqn:proof:thm:pointwiseleakagewithhdenome}, we get
\begin{align}
    \sup_{U:U-X-Y}&\log{\frac{\mathbb{E}[h(G^*(U))]}{\mathbb{E}[h(G_y^*(U))]}}\\
    &\geq \frac{\sum_{i=1}^{|\mathcal{X}|-1}h(i)P_X(\tilde{x}_i)+P_X(x^*)\tau}{\sum_{i=1}^{|\mathcal{X}|-1}h(i)P_{X|Y}(\bar{x}_i|y)+P_{X|Y}(x^*|y)\tau}\\
      &= \frac{\frac{\sum_{i=1}^{|\mathcal{X}|-1}h(i)P_X(\tilde{x}_i))}{\tau}+P_X(x^*)}{\frac{\sum_{i=1}^{|\mathcal{X}|-1}h(i)P_{X|Y}(\bar{x}_i|y)}{\tau}+P_{X|Y}(x^*|y)}\\
      &=\frac{P_X(x^*)}{P_{X|Y}(x^*|y)}\label{eqn:proof:thm:pointwiseleakagewithh3}\\
      &=D_{\infty}(P_X\|P_{X|Y=y}),
\end{align}
where \eqref{eqn:proof:thm:pointwiseleakagewithh3} follows by taking limit $m_{x^*}\rightarrow \infty$ and noting that $\tau\rightarrow \infty$ as $m_{x^*}\rightarrow \infty$ as discussed above. This completes the proof of the lower bound.

\section{Proof of Theorem~\ref{theorem:BES}}\label{proof:MLGE-BES}
 Consider an arbitrary $\mathcal{U}$. Let $\gamma(P)$ denote the guesswork of a random variable with probability distribution $P$ over $\mathcal{U}$. For a $P_{U|X}$, let $A(u)=P_{U|X}(u|0)$ and $B(u)=P_{U|X}(u|1)$, for $u\in\mathcal{U}$. Since $|\mathcal{X}|=\{0,1\}$, the optimization in \eqref{eqn:maxLviaGE} over all $P_{U|X}$ is equivalent to the optimization over the distributions $A$ and $B$. So, for the binary erasure source in \eqref{eqn:erasuresource1} and \eqref{eqn:erasuresource2}, we have
\begin{align}
    P_{U|Y}(u|0)&=A(u),\\
    P_{U|Y}(u|1)&=B(u),\\
    P_{U|Y}(u|\text{e})&=\frac{A(u)+B(u)}{2},
\end{align}
for all $u\in\mathcal{U}$. In \eqref{definition:maxLviaGE}, let $G^*(\cdot)$ and $\{G^*_y:y\in\mathcal{Y}\}$ denote the optimal guessing strategy and the optimal set of guessing strategies in the numerator and the denominator, respectively. We have
\begin{align}
    \mathbb{E}[G^*(U)]&=\gamma\left(\frac{A+B}{2}\right),\label{eqn:erasurethmnum}\\
    \mathbb{E}[G_Y^*(U)]&=\sum_{y\in\mathcal{Y}}P_Y(y)\mathbb{E}[G_y^*(U)|Y=y]\\
    &=\frac{1-p}{2}\gamma(A)+\frac{1-p}{2}\gamma(B)+p\gamma\left(\frac{A+B}{2}\right)\label{eqn:erasurethmdenom}.
\end{align}
Taking the ratio of these two quantities in \eqref{eqn:erasurethmnum} and \eqref{eqn:erasurethmdenom}, we get
\begin{align}
    \frac{\mathbb{E}[G^*(U)]}{\mathbb{E}[G_Y^*(U)]}&=\frac{\gamma\left(\frac{A+B}{2}\right)}{\frac{1-p}{2}\gamma(A)+\frac{1-p}{2}\gamma(B)+p\gamma(\frac{A+B}{2})}\\
    &=\frac{1}{(1-p)\left(\frac{\frac{1}{2}\gamma(A)+\frac{1}{2}\gamma(B)}{\gamma(\frac{A+B}{2})}\right)+p}\label{eqn:erasureproof2}.
\end{align}
Let $\mathcal{P}$ denote the set of all probability distributions on $\mathcal{U}$ of cardinality, say, $n$, i.e.,
\begin{align}
    \mathcal{P}=\{x^n\in \mathbb{R}^n:\sum_{i=1}^nx_i=1, x_i\geq 0, \ \text{for} \ i\in[1:n]\}.
\end{align}
We have
\begin{align}
       &\mathcal{L}^{G}(X\rightarrow Y)\nonumber\\
       &=\sup_{U:U-X-Y}\log{ \frac{\mathbb{E}[G^*(U)]}{\mathbb{E}[G_Y^*(U)]}}\\
       &=\sup_{P_{U|X}}\frac{1}{(1-p)\left(\frac{\frac{1}{2}\gamma(A)+\frac{1}{2}\gamma(B)}{\gamma(\frac{A+B}{2})}\right)+p}\label{eqn:erasureopt-1}\\
       &=\sup_{P_U}\sup_{A,B\in\mathcal{P}:A+B=2P_U}\frac{1}{(1-p)\left(\frac{\frac{1}{2}\gamma(A)+\frac{1}{2}\gamma(B)}{\gamma(\frac{A+B}{2})}\right)+p},\label{eqn:erasureopt}
\end{align}
where the supremum in \eqref{eqn:erasureopt-1} is decomposed into two suprema in \eqref{eqn:erasureopt} using the fact that $P_U(u)=P_{U|X}(u|0)P_X(0)+P_{U|X}(u|1)P_X(1)=\frac{A(u)+B(u)}{2}$. In \eqref{eqn:erasureopt}, it suffices to consider distributions $P_U$ over the sets of even cardinality. This is without loss of generality because, if $|\mathcal{U}|$ is odd, we can add a new realization $u$ to $\mathcal{U}$ and set $P_U(u)=0$ without changing the value of the expression in \eqref{eqn:erasureopt}.  
We compute the value of the objective function in \eqref{eqn:erasureopt} for uniform distribution $P_U$ with the following choice of $A$ and $B$ of disjoint supports. Let $\mathcal{U}=\{u_1,u_2,\dots,u_n\}$. where $n$ is even.
\begin{align}
    P_U(u_i)&=\frac{1}{n}, i\in[1:n],\label{eqn:choicedist1}\\
    A(u_{2i-1})&=\frac{2}{n}, i\in[1:\frac{n}{2}],\label{eqn:choicedist2}\\
    B(u_{2i})&=\frac{2}{n}, i\in[1:\frac{n}{2}].\label{eqn:choicedist3}
\end{align}
We have
\begin{align}
   \frac{1}{(1-p)\left(\frac{\frac{1}{2}\gamma(A)+\frac{1}{2}\gamma(B)}{\gamma(\frac{A+B}{2})}\right)+p}&=\frac{1}{(1-p)\frac{n+2}{2(n+1)}+p}\\
    &=\frac{2n+1}{n(p+1)+2}\\
    &\rightarrow \frac{2}{1+p}\ \text{as} \ n\rightarrow \infty.
\end{align}
This shows that 
$\mathcal{L}^G(X\rightarrow Y)\geq \frac{2}{1+p}$.

We next show  that the choice of the distributions in \eqref{eqn:choicedist1}-\eqref{eqn:choicedist3} is optimal for \eqref{eqn:erasureopt}. Towards this, we shall show the following claim (proved later), for a fixed $P_U$.  

\restatedclaim*

Now, for a fixed $P_U$, the inner optimization problem in \eqref{eqn:erasureopt} can be upper bounded as  
\begin{align}
    &\sup_{A,B\in\mathcal{P}:A+B=2P_U}\frac{1}{(1-p)\left(\frac{\frac{1}{2}\gamma(A)+\frac{1}{2}\gamma(B)}{\gamma(\frac{A+B}{2})}\right)+p}\\
    &=\frac{1}{(1-p)\left(\inf\limits_{A,B\in\mathcal{P}:A+B=2P_U}\frac{\frac{1}{2}\gamma(A)+\frac{1}{2}\gamma(B)}{\gamma(P_U)}\right)+p}\\
    &\leq\frac{1}{(1-p)\left(\frac{\sum_{i=1}^{\frac{n}{2}}i(P_U(u_{2i})+P_U(u_{2i-1}))}{\sum_{i=1}^niP_U(u_i)}\right)+p}\label{eqn:erasurethmupperbnd1}\\
    &\leq \frac{1}{(1-p)\frac{1}{2}+p}\label{eqn:erasureproof4}\\
    &=\frac{2}{1+p},
\end{align}
where \eqref{eqn:erasurethmupperbnd1} follows from Claim~\ref{claim} and \eqref{eqn:erasureproof4} holds because
\begin{align}\label{eqn:erasureproof3}
\frac{\sum_{i=1}^{\frac{n}{2}}i(P_U(u_{2i})+P_U(u_{2i-1}))}{\sum_{i=1}^niP_U(u_i)}\geq \frac{1}{2}.
\end{align}

It now remains to prove Claim~\ref{claim}.

\emph{Proof of Claim~\ref{claim}:} 
We first extend the definition of guesswork $\gamma(P)$ to include any real vector of length $|\mathcal{U}|$ for $P$ rather than limiting to probability distributions. That is, for $\tilde{A}:\mathcal{U}\rightarrow \mathbb{R}_+$, we define
\begin{align}
    \bar{\gamma}(\tilde{A})&=\min_{G}\sum_{i=1}^nG(u_i)\tilde{A}(u_i)\label{eqn:erasuresourceclaim1}\\
    &=\sum_{i=1}^ni\tilde{A}(u_{\sigma(i)}),
\end{align}
where the minimum in \eqref{eqn:erasuresourceclaim1} over all guessing functions $G$, and $\sigma$ is a permutation on $[1:n]$ such that $\tilde{A}(u_{\sigma(i)})\geq \tilde{A}(u_{\sigma(i)})\geq\dots \geq \tilde{A}(u_{\sigma(n)})$.
Let $\mathcal{F}$ denote the set of all functions from $\mathcal{U}$ to $\mathbb{R}_+$. Note that
\begin{align}
    &\inf_{A,B\in\mathcal{P}:A+B=2P_U}\frac{1}{2}\gamma(A)+\frac{1}{2}\gamma(B)\nonumber\\
    &=\inf_{A,B\in\mathcal{P}:A+B=2P_U}\frac{1}{2}\bar{\gamma}({A})+\frac{1}{2}\bar{\gamma}({B})\\
    &\geq\inf_{\tilde{A},\tilde{B}\in\mathcal{F}:\tilde{A}+\tilde{B}=2P_U}\frac{1}{2}\bar{\gamma}(\tilde{A})+\frac{1}{2}\bar{\gamma}(\tilde{B})\label{eqn:erasuresourceclaimstep},
\end{align}
where \eqref{eqn:erasuresourceclaimstep} holds since $\mathcal{P}\subseteq \mathcal{F}$.
In view of this, to prove \eqref{eqn:erasureproof5} in Claim~\ref{claim}, it suffices to show that
\begin{align}
     &\inf_{\tilde{A},\tilde{B}\in\mathcal{F}:\tilde{A}+\tilde{B}=2P_U}\frac{1}{2}\bar{\gamma}(\tilde{A})+\frac{1}{2}\bar{\gamma}(\tilde{B})\nonumber\\
     &=\sum_{i=1}^{\frac{n}{2}}i(P_U(u_{2i})+P_U(u_{2i-1}))\label{eqn:erasuresourceclaim2}.
\end{align}
 We show \eqref{eqn:erasuresourceclaim2} in two steps. First we show that the infimum in the left-hand-side of \eqref{eqn:erasuresourceclaim2} is attained by $\tilde{A}$ and $\tilde{B}$ with disjoint supports, denoted by $\tilde{A}\perp \tilde{B}$ (in other words, $\tilde{A}(u)\neq 0 \Rightarrow \tilde{B}(u)=0$ and $\tilde{B}(u)\neq 0 \Rightarrow \tilde{A}(u)=0$), i.e.,
\begin{align}
    &\inf_{\tilde{A},\tilde{B}\in\mathcal{F}:\tilde{A}+\tilde{B}=2P_U}\frac{1}{2}\bar{\gamma}(\tilde{A})+\frac{1}{2}\bar{\gamma}(\tilde{B})\nonumber\\
     &=\inf_{\substack{\tilde{A},\tilde{B}\in\mathcal{F}:\tilde{A}\perp \tilde{B},\\ \tilde{A}+\tilde{B}=2P_U}}\frac{1}{2}\bar{\gamma}(\tilde{A})+\frac{1}{2}\bar{\gamma}(\tilde{B})\label{eqn:erasuresourceclaim3}.
\end{align}
We then show that the infimum in the right-hand-side of \eqref{eqn:erasuresourceclaim3} is attained by $\tilde{A}^*$ and $\tilde{B}^*$ given by
\begin{align}
    \tilde{A}^*(u_{2i-1})&=2P_U(u_{2i-1}), i\in[1:\frac{n}{2}]\label{eqn:claim1proof1}\\
    \tilde{A}^*(u_{2i})&=0, i\in[1:\frac{n}{2}]\\
    \tilde{B}^*(u_{2i-1})&=0, i\in[1:\frac{n}{2}]\\
    \tilde{B}^*(u_{2i})&=2P_U(u_{2i}), i\in[1:\frac{n}{2}]\label{eqn:claim1proof2}
\end{align}
with the infimum value equal to the expression in the right-hand-side of \eqref{eqn:erasuresourceclaim2}. 

To show \eqref{eqn:erasuresourceclaim3}, consider arbitrary $\tilde{A}$ and $\tilde{B}$ such that $\tilde{A}(u)+\tilde{B}(u)=2P_U(u)$, $u\in\mathcal{U}$. We argue that, for each $u^\prime\in\mathcal{U}$, we can move all the mass $2P_U(u^\prime)$ to $\tilde{A}(u^\prime)$ (resp. $\tilde{B}(u^\prime)$) if the position of $u^\prime$ in the decreasing order of the values $\tilde{A}(u)$, $u\in\mathcal{U}$ is smaller (resp. larger) than the position of $u^\prime$ in the decreasing order of the values $\tilde{B}(u)$, $u\in\mathcal{U}$ without increasing the value of the objective function in the left-hand-side of  \eqref{eqn:erasuresourceclaim3}. Let $\tilde{A}_1(u)=\tilde{A}(u)$ and $\tilde{B}_1(u)=\tilde{B}(u)$, for all $u\in\mathcal{U}$. For $k\in[2:n+1]$, we define permutations $\sigma_k$ and $\tau_k$ on $[1:n]$, the functions $\tilde{A}_k$ and $\tilde{B}_k$, and the indices $i_k,j_k\in[1:n]$, in the following iterative manner. 

\noindent For $k\in [1:n]$
\begin{enumerate}
    \item Suppose $\sigma_k$ and $\tau_k$ are the permutations on $[1:n]$ such that 
\begin{align}
    \tilde{A}_{k}(u_{\sigma_k(1)})\geq \tilde{A}_k(u_{\sigma_k(2)})\geq \dots \geq \tilde{A}_k(u_{\sigma_k(n)}),\\
    \tilde{B}_k(u_{\tau_k(1)})\geq \tilde{B}_k(u_{\tau_k(2)})\geq \dots \geq \tilde{B}_k(u_{\tau_k(n)}).
\end{align}
\item Let $i_k,j_k\in[1:n]$ be such that $\sigma_k(i_k)=k$ and $\tau_k(j_k)=k$. If $i_k\leq j_k$, we define
\begin{align}
    \tilde{A}_{k+1}(u_k)&=\tilde{A}_k(u_k)+\tilde{B}_k(u_k)=2P_U(u_k) \\
    \tilde{B}_{k+1}(u_k)&=0\label{eqn:esaruresourceclaim4}\\
    \tilde{A}_{k+1}(u_i)&=\tilde{A}_k(u_i), i\neq k \\
    \tilde{B}_{k+1}(u_i)&=\tilde{B}_k(u_i), i\neq k,\label{eqn:esaruresourceclaim5}
\end{align}
otherwise, we define
\begin{align}
    \tilde{A}_{k+1}(u_k)&=0 \\
    \tilde{B}_{k+1}(u_k)&=\tilde{A}_k(u_k)+\tilde{B}_k(u_k)=2P_U(u_k)\\
    \tilde{A}_{k+1}(u_i)&=\tilde{A}_k(u_i), i\neq k \\
    \tilde{B}_{k+1}(u_i)&=\tilde{B}_k(u_i), i\neq k.
\end{align}
\end{enumerate}
Note that $\tilde{A}_{n+1}\perp \tilde{B}_{n+1}$. Let $i_l\leq j_l$, without loss of generality. The other case is similar. We have
\begin{align}
    &\frac{1}{2}\bar{\gamma}(\tilde{A}_l)+\frac{1}{2}\bar{\gamma}(\tilde{B}_l)\nonumber\\
    &=\sum_{i\in[1:n],i\neq i_l}i\tilde{A}_l(u_{\sigma_l(i)})+i_l\tilde{A}_l(u_l)\nonumber\\
    &\hspace{12pt}+\sum_{j\in[1:n],j\neq j_l}j\tilde{B}_l(u_{\tau_l(j)})+j_l\tilde{B}_l(u_l)\label{eqn:erasuresource6}\\
    &\geq\sum_{i\in[1:n],i\neq i_l}i\tilde{A}_l(u_{\sigma_l(i)})+i_l(\tilde{A}_l(u_1)+\tilde{B}_l(u_l))\nonumber\\
    &\hspace{12pt}+\sum_{j\in[1:n],j\neq j_l}j\tilde{B}_l(u_{\tau_l(j)})+j_l\cdot0\\
    &=\sum_{i\in[1:n]}i\tilde{A}_{l+1}(u_{\sigma_{l}(i)})+\sum_{j\in[1:n]}j\tilde{B}_{l+1}(u_{\tau_{l}(j)})\\
    &\geq \frac{1}{2}\bar{\gamma}(\tilde{A}_{l+1})+\frac{1}{2}\bar{\gamma}(\tilde{B}_{l+1})\label{eqn:erasuresource7}.
\end{align}
Repeating the steps from \eqref{eqn:erasuresource6}-\eqref{eqn:erasuresource7} appropriately for $l\in[1:n]$, we get
\begin{align}
  \frac{1}{2}\bar{\gamma}(\tilde{A}_1)+\frac{1}{2}\bar{\gamma}(\tilde{B}_1)\geq \frac{1}{2}\bar{\gamma}(\tilde{A}_{n+1})+\frac{1}{2}\bar{\gamma}(\tilde{B}_{n+1}),
  \end{align}
  where $\tilde{A}_{n+1}\perp \tilde{B}_{n+1}$.
  This proves \eqref{eqn:erasuresourceclaim3}.

  Let $|\tilde{A}|=|\text{supp}(\tilde{A})|$, where $\text{supp}(\tilde{A})=|\{u\in\mathcal{U}:\tilde{A}(u)>0\}|$, for $\tilde{A}\in\mathcal{F}$. To show that the infimum in the right-hand-side of \eqref{eqn:erasuresourceclaim3} is attained by $\tilde{A}^*$ and $\tilde{B}^*$ in \eqref{eqn:claim1proof1}-\eqref{eqn:claim1proof2}, we first argue that 
\begin{align}
    &\inf_{\substack{\tilde{A},\tilde{B}\in\mathcal{F}:\tilde{A}\perp \tilde{B},\\ \tilde{A}+\tilde{B}=2P_U}}\frac{1}{2}\bar{\gamma}(\tilde{A})+\frac{1}{2}\bar{\gamma}(\tilde{B})\nonumber\\
    &=\inf_{\substack{\tilde{A},\tilde{B}\in\mathcal{F}:\tilde{A}\perp \tilde{B},\\ \tilde{A}+\tilde{B}=2P_U,\\ |\tilde{A}|-|\tilde{B}|\in \{0,1\}}}\frac{1}{2}\bar{\gamma}(\tilde{A})+\frac{1}{2}\bar{\gamma}(\tilde{B})\label{eqn:subclaim1}
\end{align}
That is, we argue that it suffices to consider $\tilde{A}$ and $\tilde{B}$ whose support sizes differ by at most 1. To see this, consider arbitrary $\tilde{A}$ and $\tilde{B}$ such that $\tilde{A}\perp \tilde{B}$ and $\tilde{A}+\tilde{B}=2P_U$. Let $\text{supp}(\tilde{A})=\{x_1,\dots,x_p\}$ and $\text{supp}(\tilde{B})=\{y_1,\dots,y_q\}$ such that $P_U(x_i)\geq P_U(x_{i+1})$, for $i\in[1:p-1]$, and $P_U(y_j)\geq P_U(y_{j+1})$, for $j\in[1:q-1]$. Without loss of generality, suppose $p\geq q$. We argue that a few $x_i$'s from $\text{supp}(\tilde{A})$ with least probability values with respect to $P_U$ can be moved to $\text{supp}(\tilde{B})$ without increasing the value of $\bar{\gamma}(\tilde{A})+\bar{\gamma}(\tilde{B})$ so that the support sizes of the modified $\tilde{A}$ and $\tilde{B}$ differ by at most $1$.  Let $p^\prime=\lceil\frac{p+q}{2} \rceil$. We define $\tilde{A}^\prime,\tilde{B}^\prime\in\mathcal{F}$ with $|\tilde{A}^\prime|=p^\prime$ and $|\tilde{B}^\prime|=q+p-p^\prime$ as
\begin{align}
    \tilde{A}^\prime(x_i)&=\tilde{A}(x_i), i\in\left[1:p^\prime\right],\\
    \tilde{A}^\prime(x_i)&=0, i\in\left[p^\prime+1:p\right],\\
    \tilde{B}^\prime(y_j)&=\tilde{B}(y_j), j\in[1:q],\\
    \tilde{B}^\prime(x_j)&=\tilde{A}(x_j), j\in\left[p^\prime+1:p\right]. 
\end{align}
Notice that $|\tilde{A}^\prime|-|\tilde{B}^\prime|=2p^\prime-p-q\in\{0,1\}$. So,
\begin{align}
    &\bar{\gamma}(\tilde{A})+\bar{\gamma}(\tilde{B})\nonumber\\
    &=\sum_{i=1}^{\left\lceil\frac{p+q}{2}\right\rceil}iP_U(x_i)+\sum_{i=\left\lceil\frac{p+q}{2}\right\rceil+1}^piP_U(x_i)+\sum_{j=1}^qjP_U(y_j)\\
    &\geq \sum_{i=1}^{\left\lceil\frac{p+q}{2}\right\rceil}iP_U(x_i)+\sum_{j=1}^qjP_U(y_j)+\sum_{j=q+1}^{\left\lfloor\frac{p+q}{2}\right\rfloor} jP_U(x_{j+\left\lceil\frac{p-q}{2}\right\rceil})\label{eqn:claimproof3}\\
    &=\sum_{i=1}^{\left\lceil\frac{p+q}{2}\right\rceil}i\tilde{A}^\prime(x_i)++\sum_{j=1}^qj\tilde{B}^\prime(y_j)+\sum_{j=q+1}^{\left\lfloor\frac{p+q}{2}\right\rfloor} j\tilde{B}^\prime(x_{j+\left\lceil\frac{p-q}{2}\right\rceil})\\
    &\geq \bar{\gamma}(\tilde{A}^\prime)+\bar{\gamma}(\tilde{B}^\prime)\label{eqn:claimproof4},
\end{align}
where \eqref{eqn:claimproof3} holds because $p\geq q$ and \eqref{eqn:claimproof4} follows from the definition of $\bar{\gamma}$ in \eqref{eqn:erasuresourceclaim1}. This proves \eqref{eqn:subclaim1}. 

We now show that the infimum in the right-hand-side of \eqref{eqn:subclaim1} is attained by $\tilde{A}^*$ and $\tilde{B}^*$ in \eqref{eqn:claim1proof1}-\eqref{eqn:claim1proof2}. Consider $\tilde{A},\tilde{B}\in\mathcal{F}$ such that $\tilde{A}\perp \tilde{B}$, $\tilde{A}+\tilde{B}=2P_U$, and $|\tilde{A}|-|\tilde{B}|\in\{0,1\}$. Suppose again $\text{supp}(\tilde{A})=\{x_1,\dots,x_p\}$ and $\text{supp}(\tilde{B})=\{y_1,\dots,y_q\}$ with $x_i$'s and $y_i$'s in non-increasing order of probabilities with respect to $P_U$ as before. We have either $p=q$ or $p=q+1$. For $r\in [1:n]$, we argue that, if $r$ is odd (resp. even), $u_r\in\text{supp}(\tilde{B})$ (resp. $u_r\in\text{supp}(\tilde{A})$) can be swapped with $u_i\in \text{supp}(\tilde{A})$ (resp. $u_i\in \text{supp}(\tilde{B})$), for some $i>r$, without increasing the value of $\bar{\gamma}(\tilde{A})+\bar{\gamma}(\tilde{B})$. Let $\tilde{A}_1^{\prime\prime}(u)=\tilde{A}(u)$ and $\tilde{B}_1^{\prime\prime}(u)=\tilde{B}(u)$, for all $u\in\mathcal{U}$. We define $\tilde{A}_r^{\prime\prime}$ and $\tilde{B}_r^{\prime\prime}$, for $r\in[2:n]$, in the following iterative manner. Set $r=1$.

\noindent While $r\leq n-1$ 
\newline
    \hspace*{6pt} Initialize $\tilde{A}_{r+1}^{\prime\prime}(u)=\tilde{B}_{r+1}^{\prime\prime}(u)=0$, for all $u\in\mathcal{U}$.\\
    \hspace*{6pt} If $r$ is even,
    \begin{itemize}
        \item if $u_r=y_{\frac{r}{2}}$, set $r=r+1$ and exit the loop.
        \item if $u_r=x_{\frac{r}{2}+1}$, define
        \begin{itemize}
            \item $\tilde{A}_{r+1}^{\prime\prime}(x_i)=\tilde{A}_r^{\prime\prime}(x_i)$, for $x_i\in\text{supp}(\tilde{A}_r^{\prime\prime})$, $i\neq \frac{r}{2}+1$,
            \item $\tilde{B}_{r+1}^{\prime\prime}(y_j)=\tilde{A}_r^{\prime\prime}(y_j)$, for $y_j\in\text{supp}(\tilde{B}_r^{\prime\prime})$, $j\neq \frac{r}{2}+1$,
            \item $\tilde{A}_{r+1}^{\prime\prime}(y_{\frac{r}{2}+1})=\tilde{B}_{r}^{\prime\prime}(y_{\frac{r}{2}+1})$, 
            \item $\tilde{B}_{r+1}^{\prime\prime}(x_{\frac{r}{2}+1})=\tilde{A}_{r}^{\prime\prime}(x_{\frac{r}{2}+1})$.
        \end{itemize}
        \item set
        \begin{align}
            \{x_1,x_2,\dots,x_p\}&=\text{supp}(\tilde{A}_{r+1}^{\prime\prime}),\\
            \{y_1,y_2,\dots,y_q\}&=\text{supp}(\tilde{B}_{r+1}^{\prime\prime}),
        \end{align}
        such that $x_i$'s and $y_j$'s are in non-increasing order of probabilities with respect to $P_U$.
    \end{itemize}
   \hspace{6pt} else if $r$ is odd,
    \begin{itemize}
        \item if $u_r=x_{\frac{r+1}{2}}$, set $r=r+1$ and exit the loop.
        \item if $u_r=y_{\frac{r+1}{2}}$, define
        \begin{itemize}
            \item $\tilde{A}_{r+1}^{\prime\prime}(x_i)=\tilde{A}_r^{\prime\prime}(x_i)$, for $x_i\in\text{supp}(\tilde{A}_r^{\prime\prime})$, $i\neq \frac{r+1}{2}$,
            \item $\tilde{B}_{r+1}^{\prime\prime}(y_j)=\tilde{A}_r^{\prime\prime}(y_j)$, for $y_j\in\text{supp}(\tilde{B}_r^{\prime\prime})$, $j\neq \frac{r+1}{2}$,
            \item $\tilde{A}_{r+1}^{\prime\prime}(y_{\frac{r+1}{2}})=\tilde{B}_{r}^{\prime\prime}(y_{\frac{r+1}{2}})$, 
            \item $\tilde{B}_{r+1}^{\prime\prime}(x_{\frac{r+1}{2}})=\tilde{A}_{r}^{\prime\prime}(x_{\frac{r+1}{2}})$.
        \end{itemize}
        \item set
        \begin{align}
            \{x_1,x_2,\dots,x_p\}&=\text{supp}(\tilde{A}_{r+1}^{\prime\prime}),\\
            \{y_1,y_2,\dots,y_q\}&=\text{supp}(\tilde{B}_{r+1}^{\prime\prime}),
        \end{align}
        such that $x_i$'s and $y_j$'s are in non-increasing order of probabilities with respect to $P_U$.
    \end{itemize}
We now show that the value of the function $\bar{\gamma}(\tilde{A}_r)+\bar{\gamma}(\tilde{B}_r)$ does not increase in each iteration of the while loop.
    \begin{align}
        &\bar{\gamma}(\tilde{A}_r^{\prime\prime})+\bar{\gamma}(\tilde{B}_r^{\prime\prime})\nonumber\\
        &=\sum_{i=1}^p i\tilde{A}_r^{\prime\prime}(x_i)+\sum_{j=1}^q j\tilde{B}_r^{\prime\prime}(y_j)\\
        &=\sum_{i=1}^\frac{r}{2} i\tilde{A}_r(x_i)+(\frac{r}{2}+1)\tilde{B}_r(y_{\frac{r}{2}+1})+\sum_{i=\frac{r}{2}+2}^p i\tilde{A}_r(x_i)\nonumber\\
        &\hspace{12pt}+\sum_{j=1}^\frac{r}{2} j\tilde{B}_r(y_j)+(\frac{r}{2}+1)\tilde{A}_r(y_{\frac{r}{2}+1})+\sum_{j=\frac{r}{2}+2}^q j\tilde{A}_r(y_j)\\
        &\geq \bar{\gamma}(\tilde{A}_{r+1}^{\prime\prime})+\bar{\gamma}(\tilde{B}_{r+1}^{\prime\prime})\label{eqn:subclaimproof5},
    \end{align}
    where \eqref{eqn:subclaimproof5} follows from the definition of $\bar{\gamma}$ in \eqref{eqn:erasuresourceclaim1}.

Now noting that the resulting $\tilde{A}$ and $\tilde{B}$ from the implementation of the while loop above are exactly equal to $\tilde{A}^*$ and $\tilde{B}^*$ in \eqref{eqn:claim1proof1}-\eqref{eqn:claim1proof2} shows that they achieve the infimum in the right-hand-side of \eqref{eqn:erasuresourceclaim3}. This completes the proof of Claim~\ref{claim} and hence the proof of Theorem~\ref{theorem:BES}.
\section{Proof of Theorem~\ref{thm:obmgl}}\label{proof:obgml}
We first state a lemma which will be useful in the proofs of Theorems~\ref{thm:obmgl} and \ref{thm:pwrhoguesleak}.
\begin{lemma}[{\cite[Lemma~2 and (28)]{Slamatianetal19}}]\label{lemma:guessmom}
    For a joint probability distribution $P_{XY}$ and any $\rho>0$, 
    \begin{align}
    \inf_{P_{\hat{X}}}\log\mathbb{E}[V_\rho(X,\hat{X}_1^\infty)]&=\rho H_{\frac{1}{1+\rho}}(X),\\
        \inf_{P_{\hat{X}|Y}}\log\sum_{y\in\mathcal{Y}}P_Y(y)\mathbb{E}[V_\rho(X,\hat{X}_1^\infty)|Y=y]&=\rho H^\text{A}_{\frac{1}{1+\rho}}(X|Y).
    \end{align}
\end{lemma}
We prove Theorem~\ref{thm:obmgl} now.  The oblivious maximal $\rho$-guesswork leakage can be simplified as
  \begin{align}
      &\mathcal{L}_\rho^{\text{oblv}-G}(X\rightarrow Y)\nonumber\\
       &=\sup_{U:U-X-Y}\log{\frac{\inf_{P_{\hat{X}}}\mathbb{E}[V_\rho(X,\hat{X}_1^\infty)]}{\inf_{P_{\hat{X}|Y}}\sum_{y\in\mathcal{Y}}P_Y(y)\mathbb{E}[V_\rho(X,\hat{X}_1^\infty)|Y=y]}}\\
       &=\sup_{U:U-X-Y} \log{\frac{\mathsf{e}^{\rho H_{\frac{1}{1+\rho}}(U)}}{\mathsf{e}^{\rho H_{\frac{1}{1+\rho}}^{\text{A}}(U|Y)}}}\label{eqn:thm:obmgl1}\\
        &=\rho\sup_{U:U-X-Y} I_{\frac{1}{1+\rho}}^\text{A}(U;Y)\label{eqn:eqn:thm:obmgl3}\\
        &=\rho\sup_{P_{\tilde{X}}\ll P_X} I_{\frac{1}{1+\rho}}^\text{A}(\tilde{X};Y)\label{eqn:eqn:thm:obmgl2},
  \end{align}
  where \eqref{eqn:thm:obmgl1} follows from Lemma~\ref{lemma:guessmom} and \eqref{eqn:eqn:thm:obmgl2} holds because the optimization problem in \eqref{eqn:eqn:thm:obmgl3} is shown to be equal to that of in \eqref{eqn:eqn:thm:obmgl2} in \cite[Theorem~5]{LiaoSKC20}.

   \section{Proof of Theorem~\ref{thm:pwrhoguesleak}}\label{proof:pwrhoguesleak}
  Let $\alpha=\frac{1}{1+\rho}$. The pointwise oblivious maximal $\rho$-guesswork leakage can be simplified as
  \begin{align}
       &\mathcal{L}_\rho^{\text{pw-oblv}-G}(X\rightarrow y)\nonumber\\
       &=\sup_{U:U-X-Y}\log{\frac{\inf_{P_{\hat{X}}}\mathbb{E}[V_\rho(X,\hat{X}_1^\infty)]}{\inf_{P_{\hat{X}|Y=y}}\mathbb{E}[V_\rho(X,\hat{X}_1^\infty)|Y=y]}}\\
       &=\sup_{U:U-X-Y} \log{\frac{\mathsf{e}^{\rho H_{\frac{1}{1+\rho}}(U)}}{\mathsf{e}^{\rho H_{\frac{1}{1+\rho}}^{\text{A}}(U|Y=y)}}}\label{eqn:thm5proof1}\\
       &= \rho\sup_{U:U-X-Y} H_\alpha(U)-H_\alpha(U|Y=y)\\
       &=\sup_{U:U-X-Y}\log\frac{\left(\sum_{u\in\mathcal{U}}P_U(u)^\alpha\right)^\frac{1}{\alpha}}{\left(\sum_{u\in\mathcal{U}}P_{U|Y}(u|y)^\alpha\right)^\frac{1}{\alpha}},  
  \end{align}
  where \eqref{eqn:thm5proof1} follows from Lemma~\ref{lemma:guessmom}. Now we show that
  \begin{align}
      \sup_{U:U-X-Y}\frac{\left(\sum_{u\in\mathcal{U}}P_U(u)^\alpha\right)^\frac{1}{\alpha}}{\left(\sum_{u\in\mathcal{U}}P_{U|Y}(u|y)^\alpha\right)^\frac{1}{\alpha}}=\max_{x\in\mathcal{X}}\frac{P_X(x)}{P_{X|Y}(x|y)}.
  \end{align}
  We first prove the upper bound. Assume, without loss of generality, that $P_X(x)>0$, for all $x\in\mathcal{X}$. Consider 
  \begin{align}
      &\left(\sum_{u\in\mathcal{U}}P_{U|Y}(u|y)^\alpha\right)^\frac{1}{\alpha}\nonumber\\
      &=\left(\sum_{u\in\mathcal{U}}\left(\sum_{x\in\mathcal{X}}P_{UX|Y}(u,x|y)\right)^\alpha\right)^\frac{1}{\alpha}\\
      &=\left(\sum_{u\in\mathcal{U}}\left(\sum_{x\in\mathcal{X}}P_{U|X}(u|x)P_{X|Y}(x|y)\right)^\alpha\right)^\frac{1}{\alpha}\\
      &=\left(\sum_{u\in\mathcal{U}}\left(\sum_{x\in\mathcal{X}}P_{U|X}(u|x)\frac{P_{Y|X}(y|x)P_X(x)}{P_Y(y)}\right)^\alpha\right)^\frac{1}{\alpha}\\
      &\geq \left(\min_{x\in\mathcal{X}}\frac{P_{Y|X}(y|x)}{P_Y(y)}\right) \!\left(\sum_{u\in\mathcal{U}}\left(\sum_{x\in\mathcal{X}}P_{U|X}(u|x)P_X(x)\right)^\alpha\right)^\frac{1}{\alpha}\\
      &=\left(\min_{x\in\mathcal{X}}\frac{P_{Y|X}(y|x)}{P_Y(y)}\right)\left(\sum_{u\in\mathcal{U}}P_U(u)^\alpha\right)^\frac{1}{\alpha}.
  \end{align}
  So, we have
  \begin{align}
      \frac{\left(\sum_{u\in\mathcal{U}}P_U(u)^\alpha\right)^\frac{1}{\alpha}}{\left(\sum_{u\in\mathcal{U}}P_{U|Y}(u|y)^\alpha\right)^\frac{1}{\alpha}} &\leq \frac{1}{\left(\min_{x\in\mathcal{X}}\frac{P_{Y|X}(y|x)}{P_Y(y)}\right)}\\
      &=\left(\max_{x\in\mathcal{X}}\frac{P_{Y}(y)}{P_{Y|X}(y|x)}\right)\\
      &=\left(\max_{x\in\mathcal{X}}\frac{P_{X}(x)}{P_{X|Y}(x|y)}\right).
  \end{align}
We prove the lower bound now. we use the `shattering' conditional distribution $P_{U|X}$~\cite[Proof of Theorem~1]{IssaWK20},\cite[Proof of Theorem~5]{LiaoSKC20}. Let $\mathcal{U}=\cup_{x\in\mathcal{X}}\mathcal{U}_x$ (a disjoint union) and $|\mathcal{U}_x|=m_x$, for $x\in\mathcal{X}$. Define
\begin{align}
P_{U|X}(u|x)=\begin{cases}\frac{1}{m_x}, & u\in\mathcal{U}_x\\
0, & \text{otherwise}.
\end{cases}
\end{align}
This gives
\begin{align}
    P_U(u)&=\frac{P_X(x)}{m_x}, u\in\mathcal{U}_x\\
    P_{U|Y}(u|y)&=\frac{P_{X|Y}(x|y)}{m_x}, u\in\mathcal{U}_x.
\end{align}
So, we have
\begin{align}
     &\sup_{U:U-X-Y}\log\frac{\left(\sum_{u\in\mathcal{U}}P_U(u)^\alpha\right)^\frac{1}{\alpha}}{\left(\sum_{u\in\mathcal{U}}P_{U|Y}(u|y)^\alpha\right)^\frac{1}{\alpha}}\nonumber\\
     &\geq \sup_{m_x, x\in\mathcal{X}}\frac{\left(\sum_{x\in\mathcal{X}}\sum_{u\in\mathcal{U}_x}\frac{P_X(x)^\alpha}{m_x^\alpha}\right)^\frac{1}{\alpha}}{\left(\sum_{x\in\mathcal{X}}\sum_{u\in\mathcal{U}_x}\frac{P_{X|Y}(x|y)^\alpha}{m_x^\alpha}\right)^\frac{1}{\alpha}}\\
     &= \sup_{m_x, x\in\mathcal{X}}\frac{\left(\sum_{x\in\mathcal{X}}\sum_{u\in\mathcal{U}_x}\frac{P_Y(y)^\alpha}{P_{Y|X}(y|x)^\alpha}\frac{P_{X|Y}(x|y)^\alpha}{m_x^\alpha}\right)^\frac{1}{\alpha}}{\left(\sum_{x\in\mathcal{X}}\sum_{u\in\mathcal{U}_x}\frac{P_{X|Y}(x|y)^\alpha}{m_x^\alpha}\right)^\frac{1}{\alpha}}\\
     &= \sup_{m_x, x\in\mathcal{X}}\frac{\left(\sum_{x\in\mathcal{X}}\frac{P_Y(y)^\alpha}{P_{Y|X}(y|x)^\alpha}\frac{P_{X|Y}(x|y)^\alpha}{m_x^{\alpha-1}}\right)^\frac{1}{\alpha}}{\left(\sum_{x\in\mathcal{X}}\frac{P_{X|Y}(x|y)^\alpha}{m_x^{\alpha-1}}\right)^\frac{1}{\alpha}}\\
     &=\sup_{m_x, x\in\mathcal{X}}\left(\sum_{x\in\mathcal{X}}\left(\frac{P_Y(y)}{P_{Y|X}(y|x)}\right)^\alpha P_{\hat{X}}(x)\right)^\frac{1}{\alpha}\label{eqn:th5proofclaim1}\\
     &=\sup_{P_{\hat{X}}\ll P_{X|Y=y}}\left(\sum_{x\in\mathcal{X}}\left(\frac{P_Y(y)}{P_{Y|X}(y|x)}\right)^\alpha P_{\hat{X}}(x)\right)^\frac{1}{\alpha}\label{eqn:th5proofclaim2},
\end{align}
where \eqref{eqn:th5proofclaim1} follows by defining 
\begin{align}
    P_{\hat{X}}(x)=\frac{\left(\frac{P_{X|Y}(x|y)^\alpha}{m_x^{\alpha-1}}\right)}{\left(\sum_{x^\prime\in\mathcal{X}}\frac{P_{X|Y}(x^\prime|y)^\alpha}{m_{x^\prime}^{\alpha-1}}\right)}
\end{align}
and \eqref{eqn:th5proofclaim2} follows because $P_{\hat{X}}(x)$ can be made arbitrarily close to any distribution with the same support as $P_{X}$ for sufficiently large $m_x$, $x\in\mathcal{X}$, along the same lines as \cite[Equation~(60)]{LiaoSKC20}. 

Finally, note that
\begin{align}
    \sup_{P_{\hat{X}}\ll P_{X|Y=y}}\left(\sum_{x\in\mathcal{X}}\left(\frac{P_Y(y)}{P_{Y|X}(y|x)}\right)^\alpha P_{\hat{X}}(x)\right)^\frac{1}{\alpha}\!=\max_{x\in\mathcal{X}}\frac{P_X(x)}{P_{X|Y}(x|y)}.
\end{align}

This follows by noting that
\begin{align}
    &\left(\sum_{x\in\mathcal{X}}\left(\frac{P_Y(y)}{P_{Y|X}(y|x)}\right)^\alpha P_{\hat{X}}(x)\right)^\frac{1}{\alpha}\nonumber\\
    &\leq\left(\sum_{x\in\mathcal{X}}\left(\max_{x^\prime\in\mathcal{X}}\frac{P_Y(y)}{P_{Y|X}(y|x^\prime)}\right)^\alpha P_{\hat{X}}(x)\right)^\frac{1}{\alpha}\label{eqn:thm5proofclaim5}\\
    &=\left(\max_{x^\prime\in\mathcal{X}}\frac{P_Y(y)}{P_{Y|X}(y|x^\prime)}\right)\left(\sum_{x\in\mathcal{X}}P_{\hat{X}}(x)\right)^\frac{1}{\alpha}\\
    &=\left(\max_{x^\prime\in\mathcal{X}}\frac{P_{X}(x^\prime)}{P_{X|Y}(x|y)}\right)\left(\sum_{x\in\mathcal{X}}P_{\hat{X}}(x)\right)^\frac{1}{\alpha}\\
    &=\max_{x^\prime\in\mathcal{X}}\frac{P_{X}(x^\prime)}{P_{X|Y}(x|y)}
\end{align}
and that equality in \eqref{eqn:thm5proofclaim5} is attained by $P_{\hat{X}}$ such that $P_{\hat{X}}(x^*)=1$, for a fixed $x^*\in\argmax_{x\in\mathcal{X}}\frac{P_{X}(x)}{P_{X|Y}(x|y)}$. This completes the proof.

\fi

\end{document}